\documentclass[copyright,creativecommons]{eptcs}

\usepackage[english]{babel}
\usepackage[latin1]{inputenc}
\usepackage{amsmath,amssymb,amsfonts,mathrsfs,latexsym,stmaryrd}
\usepackage{array}
\usepackage{url}
\usepackage[newitem,newenum,neverdecrease]{paralist}
\usepackage{ifpdf}
\usepackage{color}
\definecolor{darkcyan}{rgb}{0,0.55,0.55}
\usepackage[OT1]{eulervm}
\usepackage{times}
\usepackage{graphicx}
\usepackage{latexsym}
\usepackage{fancybox}
\usepackage{tikz}
\usetikzlibrary{decorations,arrows,shapes}
\usepgflibrary{decorations.pathreplacing} 
\usepgflibrary[decorations.pathreplacing] 
\usetikzlibrary{decorations.pathreplacing} 
\usetikzlibrary[decorations.pathreplacing]
\usepackage{stackrel}
\usepackage{xspace}
\usepackage{amsthm}

\newtheorem{definition}{Definition}

\newtheorem{theorem}{Theorem}
\newtheorem{lemma}{Lemma}

\newcommand{\mmodels}{\Vdash}

\newcommand{\hs}{{\rm HS}\xspace}     

\newcommand{\finishes}{\ensuremath{\langle E \rangle}\xspace}
\newcommand{\during}{\ensuremath{\langle D \rangle}\xspace}
\newcommand{\starts}{\ensuremath{\langle B \rangle}\xspace}
\newcommand{\overlaps}{\ensuremath{\langle O \rangle}\xspace}
\newcommand{\meets}{\ensuremath{\langle A \rangle}\xspace}

\newcommand{\Tfinishes}{\ensuremath{\langle \overline{E} \rangle}\xspace}
\newcommand{\Tduring}{\ensuremath{\langle \overline{D} \rangle}\xspace}
\newcommand{\Tstarts}{\ensuremath{\langle \overline{B} \rangle}\xspace}
\newcommand{\Toverlaps}{\ensuremath{\langle \overline{O} \rangle}\xspace}
\newcommand{\Tmeets}{\ensuremath{\langle \overline{A} \rangle}\xspace}

\newcommand{\later}{\ensuremath{\langle L \rangle}\xspace}
\newcommand{\Tlater}{\ensuremath{\langle \overline{L} \rangle}\xspace}

\newcommand{\Gfinishes}{\ensuremath{[ E ]}\xspace}

\newcommand{\Gmeets}{\ensuremath{[ A ]}\xspace}

\newcommand{\TGmeets}{\ensuremath{[ \overline{A} ]}\xspace}

\newcommand{\Glater}{\ensuremath{[ L ]}\xspace}

\newcommand{\blank}{\ensuremath{\mathsf{\$b}}\xspace}

\newlength{\partextlength}
\newcommand{\partext}[1]{\text{\parbox[c]{\partextlength}{\em #1}}}

\allowdisplaybreaks[4]

\providecommand{\event}{GandALF 2012} 
\usepackage{breakurl}             

\title{Interval Temporal Logics over Strongly Discrete\\ Linear Orders: the Complete Picture}

\author{Davide Bresolin
\institute{University of Verona (Italy)
\email{davide.bresolin@univr.it}}
\quad\and
Dario Della Monica
\institute{Reykjavik University (Iceland)
\email{dariodm@ru.is}}
\and
Angelo Montanari
\institute{University of Udine (Italy)
\email{angelo.montanari@uniud.it}}
\and
Pietro Sala
\institute{University of Verona (Italy)
\email{pietro.sala@univr.it}}
\and
Guido Sciavicco
\institute{University of Murcia (Spain)
\email{guido@um.es}}
}
\def\titlerunning{Interval Temporal Logics over Strongly Discrete Linear Orders}
\def\authorrunning{D. Bresolin et al.}
\begin{document}

\maketitle

\begin{abstract}
Interval temporal logics provide a general framework for
temporal reasoning about interval structures over linearly
ordered domains, where intervals are taken as the primitive
ontological entities. In this paper, we identify all fragments
of Halpern and Shoham's interval temporal logic \hs\ with
a decidable satisfiability problem over the class of strongly
discrete linear orders. We classify them in terms of both their
relative expressive power and their complexity.
We show that there are exactly 44
expressively different decidable fragments, whose complexity
ranges from NP to EXPSPACE. In addition, we identify some new
undecidable fragments (all the remaining \hs\ fragments were
already known to be undecidable over strongly discrete linear
orders). We conclude the paper by an analysis of the specific case of
natural numbers, whose behavior slightly differs from that of
the whole class of strongly discrete linear orders.
The number of decidable fragments over $\mathbb N$
raises up to 47: three undecidable fragments become decidable
with a non-primitive recursive complexity.
\end{abstract}

\section{Introduction}

Interval temporal logics provide a general framework for
temporal reasoning about interval structures over linearly
(or partially) ordered domains. They take time intervals as
the primitive ontological entities and define truth of formulas
relative to time intervals, rather than time points. Interval
logic modalities correspond to various relations between pairs
of intervals, with the exception of Venema's CDT and its fragments,
that consider ternary relations~\cite{JLOGC::Venema1991}.
In particular, Halpern and Shoham's modal logic
of time intervals \hs~\cite{HalpernS91} features a set of
modalities that makes it possible to express all Allen's
interval relations~\cite{allen83} (see Table~\ref{tab:relations}).

Interval-based formalisms have been extensively used in many
areas of computer science, such as, for instance, planning,
natural language processing, constraint satisfaction, and
verification of hardware and software systems. However, most
of them impose severe syntactic and semantic restrictions that
considerably weaken their expressive power.
Interval temporal logics relax these restrictions, allowing one
to cope with much more complex application domains and scenarios.
Unfortunately, many of them, including \hs\ and the majority of
its fragments, turn out to be undecidable~\cite{lpar08}.
\begin{table}[t]
\centering
\begingroup
\newcommand{\albedoplot}{%
\begin{tikzpicture}[scale=0.7,font=\footnotesize]
\useasboundingbox (0,0) rectangle (3.5,0);
\draw[dotted,help lines,thick] (0,-.25) -- ++(0,-5.05);
\draw[dotted,help lines,thick] (2,-.25) -- ++(0,-5.05);
\draw[red,|-|,thick] (0,0) node[above=1pt](a){$x$} -- (2,0)node[above](b){$y$};
\end{tikzpicture}}
\renewcommand{\arraystretch}{1.25}
\begin{tabular}{|c|c|l|c|}
\hline
\bf Relation & \bf Operator & \,\hfill\bf Formal definition\hfill\, & \bf Pictorial example \\
\hline
										&									&		& \albedoplot \\
\em meets						&	$\meets$					& $[x,y] R_A [x',y'] \Leftrightarrow y=x'$ &
	\tikz[scale=0.7,font=\footnotesize]{\useasboundingbox (0,-.1) rectangle (3.5,0.1);
 		\draw[|-|,thick] (2,0)node[above](c){$x'$}-- ++(1,0)node[above](d){$y'$};} \\
\em before						&	$\later$					& $[x,y] R_L [x',y'] \Leftrightarrow y < x'$ &
	\tikz[scale=0.7,font=\footnotesize]{\useasboundingbox (0,-.1) rectangle (3.5,0.1);
 		\draw[|-|,thick] (2.5,0)node[above](c){$x'$}-- ++(1,0)node[above](d){$y'$};} \\
\em started-by				&	$\starts$				& $[x,y] R_B [x',y'] \Leftrightarrow x=x', y' < y$ &
	\tikz[scale=0.7,font=\footnotesize]{\useasboundingbox (0,-.1) rectangle (3.5,0.1);
 		\draw[|-|,thick] (0,0)node[above](c){$x'$}-- ++(0.75,0)node[above](d){$y'$};} \\
\em finished-by			&	$\finishes$			& $[x,y] R_E [x',y'] \Leftrightarrow y=y', x < x'$ &
	\tikz[scale=0.7,font=\footnotesize]{\useasboundingbox (0,-.1) rectangle (3.5,0.1);
 		\draw[|-|,thick] (1.25,0)node[above](c){$x'$}-- ++(0.75,0)node[above](d){$y'$};} \\
\em contains					&	$\during$				& $[x,y] R_D [x',y'] \Leftrightarrow x < x', y' < y$ &
	\tikz[scale=0.7,font=\footnotesize]{\useasboundingbox (0,-.1) rectangle (3.5,0.1);
 		\draw[|-|,thick] (0.5,0)node[above](c){$x'$}-- ++(1,0)node[above](d){$y'$};} \\
\em overlaps					&	$\overlaps$			& $[x,y] R_O [x',y'] \Leftrightarrow x < x' < y < y'$ &
	\tikz[scale=0.7,font=\footnotesize]{\useasboundingbox (0,-.1) rectangle (3.5,0.1);
 		\draw[|-|,thick] (1,0)node[above](c){$x'$}-- ++(2,0)node[above](d){$y'$};} \\
\hline
\end{tabular}
\endgroup
\caption{Allen's interval relations and the corresponding~\hs\ modalities.}
\label{tab:relations}
\end{table}

In this paper, we focus our attention on the class of strongly discrete linear orders,
that is, of those linear structures characterized by the presence of finitely
many points in between any two points. This class includes, for instance, $\mathbb N$,
$\mathbb Z$, and all finite linear orders.
%
%
We give a complete classification of all \hs\ fragments
(defined by restricting the set of modalities), reviewing known
results and solving open problems; the results differ, as we will see,
from those in the class of all finite linearly ordered sets~\cite{ecai2012}.
The aim of such a classification is twofold: on the one hand, we identify the subset of
all expressively-different decidable fragments, thus marking the decidability border;
on the other hand, we determine the exact complexity of each of them.
As shown in Figure~\ref{fig:expre}, $\mathsf{A\overline AB\overline B}$ (that
features modal operators for Allen's relations {\em meets} and {\em started-by}, and
their inverses) and its mirror image $\mathsf{A\overline AE\overline E}$ (that replaces
relations {\em starts} and {\em started-by} by relations {\em finishes} and
{\em finished-by}) are the minimal fragments including all decidable subsets
of operators from the \hs repository, for a total of 62 languages. Of those,
44 turn out to be decidable.
As a matter of fact, the status of various fragments was already known:
\begin{inparaenum}[\it (i)]
\item $\mathsf{D}$, $\mathsf{\overline D}$, $\mathsf{O}$, and
$\mathsf{\overline O}$ have been shown to be undecidable
in~\cite{DBLP:conf/time/BresolinMGMS11,DBLP:conf/lics/MarcinkowskiM11};
\item $\mathsf{BE}$, $\mathsf{B\overline E}$, $\mathsf{\overline BE}$, and
$\mathsf{\overline B\overline E}$ are undecidable, as they can define, respectively,
$\during$ (by the equation $\during p \equiv \starts \finishes p$),
$\Toverlaps$ ($\Toverlaps p \equiv \starts \Tfinishes p$),
$\overlaps$ ($\overlaps p \equiv \finishes \Tstarts p$),
and $\Tduring$ ($\Tduring p \equiv \Tstarts \Tfinishes p$);
\item undecidability of $\mathsf{A \overline A \overline B}$
(resp., $\mathsf{A\overline A\overline E}$) can be shown using the same
technique used in~\cite{DBLP:conf/icalp/MontanariPS10} to prove the
undecidability of $\mathsf{A\overline AB}$ (resp., $\mathsf{A\overline AE}$);
\item $\mathsf{AB\overline B\overline L}$ (resp., $\mathsf{\overline
AE\overline EL}$) is in EXPSPACE~\cite{ijfcs2012}, and the
proof of EXPSPACE-hardness for $\mathsf{AB}$ and $\mathsf{A \overline B}$
(resp., $\mathsf{\overline AE}$ and $\mathsf{\overline A \overline E}$)
over finite linear orders~\cite{ecai2012} can be easily adapted to
the case of strongly discrete linear orders;
\item $\mathsf{A\overline A}$ (a.k.a. Propositional Neighborhood Logic)
is in NEXPTIME~\cite{DBLP:conf/tableaux/BresolinMSS11,Goranko03a},
and NEXPTIME-hardness already holds for $\mathsf A$ and $\mathsf{\overline
A}$~\cite{bresolin07b};
\item $\mathsf{B\overline B}$ is NP-complete~\cite{Goranko04},
and, obviously, NP-hardness already holds for $\mathsf B$ and
$\mathsf{\overline B}$ (both include propositional logic);
\item the relative expressive power of the \hs\ fragments we are
interested in is as shown in Figure~\ref{fig:expre},
whose soundness and completeness follow from the results given
in~\cite{DBLP:conf/ijcai/MonicaGMS11} and in~\cite{ecai2012},
respectively, as definability (resp., undefinability) results
transfer from more (resp., less) general to less (resp.,
more) general classes.
\end{inparaenum}

In this paper, we complete the picture by proving the following new results:
\begin{inparaenum}[\it (i)]
\item the undecidability of $\mathsf{A\overline AB}$ (resp., $\mathsf{A\overline AE}$)
and $\mathsf{A\overline A \overline B}$ (resp., $\mathsf{A\overline A\overline E}$)
can be sharpened to $\mathsf{\overline AB}$ (resp., $\mathsf{AE}$) and
$\mathsf{\overline A\overline B}$ (resp., $\mathsf{A\overline E}$), respectively
(Section 3);
\item the NP-completeness (in particular, NP-membership) of $\mathsf{B\overline B}$
can be extended to $\mathsf{B\overline BL\overline L}$ (Section \ref{sec:NP}).
\end{inparaenum}
In addition, we analyze the behavior of the various fragments over interesting
sub-classes of the class of all strongly discrete linearly ordered sets, taking
as an example that of models based on $\mathbb N$ (Section 6).
As $\mathbb N$-models are not left/right symmetric, reversing the time order
and coherently replacing modalities (e.g., $\meets$ by $\Tmeets$) does not preserve,
in general, the computational properties of a fragment. We show that:
\begin{inparaenum}[\it (i)]
\item $\mathsf{\overline AB}$ becomes decidable (which is a direct consequence
of~\cite{DBLP:conf/icalp/MontanariPS10}), precisely, non-primitive recursive
\cite{ecai2012};
\item the same holds for $\mathsf{\overline A\overline B}$ and $\mathsf{\overline
AB\overline B}$, but, in these cases, the decidability proof for $\mathsf{A\overline A
B\overline B}$ given in~\cite{DBLP:conf/icalp/MontanariPS10} must be suitably adapted;
\item $\mathsf{\overline ABL}$, $\mathsf{\overline A\overline BL}$, and $\mathsf{\overline
AB\overline BL}$ remain undecidable, but the original reductions must be
suitably adapted.
\end{inparaenum}
Thus, the number of decidable fragments over $\mathbb N$ raises up to 47, the
three new decidable fragments being all non-primitive recursive.
%
%
In fact, we can slightly generalize such a result, as the addition of finite
linear orders (finite prefixes of $\mathbb N$) to $\mathbb N$ does not alter the
%
%
picture; however, to keep presentation and proofs as simple as possible, we
restrict our attention to $\mathbb N$-models only. Symmetric results can
be obtained in the case of negative integers.

\section{\hs and its Fragments}
\label{sec:pre}

Let $\mathbb{D} = \langle D, <\rangle$ be a \emph{strongly discrete linearly
ordered set}, that is, a linearly ordered set where for every pair $x,y$, with
$x<y$, there exist at most finitely many $z_1,z_2,\ldots,z_n$
such that $x<z_1<z_2<\ldots<z_n<y$.
%
%
According to the strict approach, we exclude intervals with coincident
endpoints (point-intervals) from the semantics: an \emph{interval} over
$\mathbb{D}$ is an ordered pair $[x,y]$, with $x,y \in D$ and $x < y$.

12 different ordering relations (plus equality) between any pair of intervals
are possible, often called \emph{Allen's relations}~\cite{allen83}:
the six relations depicted in Table~\ref{tab:relations} and their inverses.
We interpret interval structures as Kripke structures and Allen's relations
as accessibility relations, thus associating a modality $\langle X\rangle$
with each Allen's relation $R_{X}$.
For each modality $\langle X \rangle$, its \emph{inverse} (or {\em transpose}),
denoted by $\langle \overline{X} \rangle$, corresponds to the inverse relation
$R_{\overline{X}}$ of $R_{X}$ (that is, $R_{\overline{X}} = (R_{X})^{-1}$).
Halpern and Shoham's logic \hs\ is a multi-modal logic whose formulas are built
on a set $\mathcal{AP}$ of proposition letters, the boolean connectives
$\vee$ and $\neg$, and one modality for each Allen's relation.
We associate a fragment $\mathsf{X_1 X_2 \ldots X_k}$ of \hs with every
subset $\{R_{X_1}, \ldots, R_{X_k}\}$ of Allen's relations, whose formulas
are defined by the following grammar:
$$ \varphi ::= p \mid
\neg \varphi \mid \varphi \vee \varphi \mid \langle
X_1\rangle\varphi\mid\ldots\mid\langle X_k\rangle\varphi.
$$
The other boolean connectives can be viewed as abbreviations, and the dual
operators $[X]$ are defined as usual ($[X] \varphi \equiv \neg
\langle X \rangle \neg \varphi$). Given a formula $\varphi$, its {\em length}
$|\varphi|$ is the number of its symbols.

The semantics of \hs\ is given in terms of {\em interval models}  $M
= \langle\mathbb{I(D)},V\rangle$, where $\mathbb{I(D)}$ is the set of all
intervals over $\mathbb{D}$. The \emph{valuation function} $V : \mathcal{AP} \mapsto 2^{\mathbb{I(D)}}$
assigns to every $p \in \mathcal{AP}$ the set of intervals $V(p)$
over which $p$ holds. The {\em truth} of a formula over a given interval
$[x,y]$ of an interval model $M$ is defined by structural induction on
formulas:
\medskip
\begin{compactitem}
\item $M,[x,y] \mmodels p$ iff $[x,y]\in V(p)$, for all $p \in
      \mathcal{AP}$;
\item $M,[x,y] \mmodels \neg\psi$ iff it is not the
      case that $M,[x,y]\mmodels\psi$;
\item $M,[x,y]\mmodels \varphi \vee \psi $ iff
      $M,[x,y]\mmodels\varphi$ or
      $M,[x,y]\mmodels\psi$;
\item $M,[x,y] \mmodels \langle X\rangle\psi$ iff there exists an interval
        $[x',y']$ such that $[x,y]R_{X}[x',y']$ and $M,[x',y']\mmodels\psi$,
        where $R_{X}$ is the relation corresponding to $\langle X\rangle$.
\end{compactitem}
\medskip
An \hs-formula $\phi$ is {\em valid}, denoted by $\mmodels \phi$, if it
is true over every interval of every interval model.

\begin{figure*}[tbp]
   \centering
		\input{disc_dec_fragm}
   \caption{Hasse diagram of fragments of $\mathsf{A\overline
   AB\overline B}$ and $\mathsf{A\overline AE\overline E}$
   over strongly discrete linear orders.}\label{fig:expre}
\end{figure*}

In this paper, we study expressiveness and computational complexity of \hs\
fragments over the class of strongly discrete linear orders.
Given a fragment $\mathcal F = \mathsf{X_1 X_2 \ldots X_k}$ and a modality
$\langle X \rangle$, we write $\langle X \rangle \in \mathcal F$ if $X \in \{X_1,
\ldots, X_k\}$. Given two fragments  $\mathcal F_1$ and $\mathcal F_2$, we write
$\mathcal F_1 \subseteq \mathcal F_2$ if $\langle X \rangle \in \mathcal F_1$
implies $\langle X \rangle \in \mathcal F_2$, for every modality $\langle
X \rangle$.

\begin{definition}
We say that an \hs\ modality $\langle X\rangle$ is \emph{definable} in an
\hs\ fragment $\mathcal F$ if there exists a formula $\psi(p) \in \mathcal F$
such that $\langle X\rangle p \leftrightarrow \psi(p)$ is valid, for any fixed
proposition letter $p$. In such a case, the equivalence $\langle X \rangle p
\equiv \psi(p)$ is called an \emph{inter-definability equation for $\langle X
\rangle$ in $\mathcal F$}.
\end{definition}

\begin{definition}
Let $\mathcal F_1$ and $\mathcal F_2$ be two \hs\ fragments. We say
that
\begin{inparaenum}[\it (i)]
\item $\mathcal F_2$ is \emph{at least as expressive as} $\mathcal F_1$
($\mathcal F_1 \preceq \mathcal F_2$) if every modality $ \langle
X\rangle \in \mathcal F_1$ is definable in $\mathcal F_2$;
\item $\mathcal F_1$ is \emph{strictly less expressive} than $\mathcal F_2$,
($\mathcal F_1 \prec \mathcal F_2$) if $\mathcal F_1 \preceq \mathcal F_2$,
but not $\mathcal F_2 \preceq \mathcal F_1$;
\item  $\mathcal F_1$ and $\mathcal F_2$ are \emph{equally expressive}, or
\emph{expressively equivalent} ($\mathcal F_1 \equiv \mathcal F_2$),
if $\mathcal F_1 \preceq \mathcal F_2$ and $\mathcal F_2 \preceq \mathcal F_1$;
\item  $\mathcal F_1$ and $\mathcal F_2$ are \emph{expressively incomparable}
($\mathcal F_1 \not\equiv \mathcal F_2$) if neither $\mathcal F_1
\preceq \mathcal F_2$ nor $\mathcal F_2 \preceq \mathcal F_1$.
\end{inparaenum}
\end{definition}

We denote each \hs\ fragment $\mathcal F$ by the list of its modalities in alphabetical
order, omitting those modalities which are definable in terms of the others.
As a matter of fact, in our setting, only $\later $ and $\Tlater$
turn out to be definable in some fragments. Any fragment $\mathcal F$ can be
transformed into its mirror image by reversing the time order and simultaneously
replacing (each occurrence of) $\meets$ by $\Tmeets$, $\later$ by $\Tlater$,
$\starts$ by $\finishes$, and $\Tstarts$ by $\Tfinishes$.
In the considered class of linear orders, the mirroring operation can be applied
to any fragment preserving all its computational properties. Thus, all results
given in this paper, except for the ones in Section 6, hold both for the considered
fragments and their mirror images. When the considered class of
models is not left/right symmetric, as it happens with $\mathbb N$ (Section 6), this is no
longer true. The rest of the paper, with the exception of Section 6,
is devoted to prove the following theorem.
\begin{theorem}
The Hasse diagram in Figure~\ref{fig:expre} correctly shows all the decidable fragments of
\hs\ over the class of strongly discrete linear orders, their relative expressive power,
and the precise complexity class of their satisfiability problem.
\end{theorem}

\section{Relative Expressive Power and Undecidability}\label{sec:expr_undec}

The most basic definability results in \hs, e.g., \hs $\equiv$ $\mathsf{A\overline
AB\overline BE\overline E}$, are known since~\cite{HalpernS91}.
In order to show that a given modality is not definable in a specific \hs\
fragment, we make use of the standard notion of bisimulation and the invariance
of modal formulas with respect to bisimulations (see, e.g., \cite{modal-logic}).
In particular, we exploit the fact that, given a modal logic $\mathcal F$,
any $\mathcal F$-bisimulation preserves the truth of all formulas in
$\mathcal F$.
Thus, in order to prove that a modality $ \langle X\rangle$ is not definable in
$\mathcal F$, it suffices to construct a pair of interval models $M$ and $M'$
and an $\mathcal F$-bisimulation between them that relates a pair of intervals
$[x,y] \in M$ and  $[x',y'] \in M'$ such that $M,[x,y] \mmodels  \langle X
\rangle p$ and $M',[x',y'] \not\mmodels  \langle X\rangle p$.

In the following, in order to prove that Figure~\ref{fig:expre} is sound and
complete for the class of all strongly discrete linear orders, we focus our
attention on fragments of $\mathsf{A\overline AB\overline B}$ and of its
mirror image $\mathsf{A\overline AE\overline E}$, and we show that the
set of nodes of the graph in Figure~\ref{fig:expre} is the set of all expressively
different fragments of $\mathsf{A\overline AB\overline B}$ and $\mathsf{A\overline
AE\overline E}$ (including $\mathsf{A\overline AB\overline B}$ and
$\mathsf{A\overline AE\overline E}$ themselves).
Nodes are partitioned with respect to the complexity of their satisfiability problem:
nodes corresponding to undecidable fragments are identified by a red rectangle
and by the superscript $1$, while nodes corresponding to EXPSPACE-complete (resp.,
NEXPTIME-complete, NP-complete) fragments are identified by a yellow rectangle
and the superscript $2$ (resp., blue rectangle/superscript $3$, green
rectangle/superscript $4$).
All \hs fragments that do not appear in the picture are undecidable.
Graph edges represent the relative expressive power of two fragments:
if two nodes, labeled by the fragments $\mathcal F_1$ and $\mathcal F_2$,
are connected by a path going from $\mathcal F_1$ to $\mathcal F_2$, then
$\mathcal F_2 \prec \mathcal F_1$;
%
%
if two fragments $\mathcal F_1$ and $\mathcal F_2$ are not connected
by a path, then $\mathcal F_1 \not\equiv \mathcal F_2$.
%
%
Thus, to show that Figure~\ref{fig:expre} is sound and complete, we need to
prove that \begin{inparaenum}[\it (i)] \item each fragment $\mathcal F_1$
connected to a fragment $\mathcal F_2$ by an arrow is strictly more expressive
than $\mathcal F_2$; \item pairs of fragments in Figure~\ref{fig:expre}, which
are not connected by a path, are expressively incomparable; and
\item the complexity of the satisfiability problem for the considered fragments
is correctly depicted in Figure~\ref{fig:expre}. \end{inparaenum}
Conditions \textit{(i)} and \textit{(ii)} are direct consequences of the
following lemma,  whose proof, given in ~\cite{ecai2012}, makes use of
bisimulations based on finite linearly ordered sets. As the class of all
strongly discrete linearly ordered sets includes that of finite linearly
ordered sets, all results immediately apply.
\begin{lemma}[\cite{ecai2012}]\label{lem:definability-equations}
The only definability equations for the \hs\ fragment
$\mathsf{A\overline AB\overline B}$, over the class
of all strongly discrete linear orders, are $\later p \equiv\meets
\meets p$ and $\Tlater p \equiv\Tmeets \Tmeets p$.
\end{lemma}
Hence, we can restrict our attention to condition \textit{(iii)}. The rest of
the section is devoted to prove the undecidability of all fragments marked as
undecidable in Figure~\ref{fig:expre}. All fragments which are not referred
to in the figure have already been proved undecidable over the class of strongly
discrete linearly ordered
sets~\cite{DBLP:conf/time/BresolinMGMS11,DBLP:conf/lics/MarcinkowskiM11}.
All decidable fragments of \hs\ over the class of strongly discrete linear
orders are thus depicted in Figure~\ref{fig:expre}. Section \ref{sec:NP}
and \ref{sec:NEXPEXSP} will be devoted to the identification of the exact
complexity of these decidable fragments.

The undecidability result we give here resembles those
in~\cite{ecai2012,DBLP:conf/icalp/MontanariPS10}. Nevertheless,
the required modifications are far from being trivial.
From~\cite{DBLP:conf/icalp/MontanariPS10,sala}, we know that there exists
a reduction from the structural termination problem for lossy counter
automata, which is known to be undecidable~\cite{lossy}, to the satisfiability
problem for $\mathsf{A\overline AB}$ and $\mathsf{A \overline A\overline B}$.
Here, we consider the nonemptiness problem for incrementing counter automata
over infinite words, which is known to be undecidable~\cite{Demri:2006:LFQ:1157735.1158038},
and we show that it can be reduced to the satisfiability problem for the fragments
$\mathsf{\overline AB}$, $\mathsf{\overline A\overline B}$, $\mathsf{A E}$,
and $\mathsf{A\overline E}$. For the sake of brevity, we will work out all
the details of the reduction for $\mathsf{AE}$ only. Since
$\mathsf{AE}$ and $\mathsf{\overline AB}$ are completely symmetric with
respect to the class of strongly discrete linearly ordered sets, the
reduction for $\mathsf{AE}$ basically works for $\mathsf{\overline AB}$
as well. Moreover, adapting it to $\mathsf{A\overline E}$ (and therefore,
by symmetry, to $\mathsf{\overline A\overline B}$) is straightforward.
Incrementing counter automata can be viewed as a variant of lossy counter
automata where faulty transitions increase the values of counters instead
of decrementing them. Hence, some of the basic concepts
of the reduction given in \cite{DBLP:conf/icalp/MontanariPS10,sala}
can be exploited.
A comprehensive survey on faulty machines and on the relevant complexity, decidability,
and undecidability results can be found in~\cite{DBLP:conf/stacs/BouyerMOSW08}.
Formally, an {\em incrementing counter automaton} is a tuple $\mathcal A
= (\Sigma,Q,q_0,C,\Delta,F)$, where $\Sigma$ is a finite alphabet, $Q$ is a
finite set of control {\em states}, $q_0\in Q$ is the initial state, $C =
\{c_1, \ldots, c_k\}$ is the set of {\em counters}, whose values range over
$\mathbb N$, $\Delta$ is a {\em transition relation}, and $F\subseteq Q$
is the set of final states.
Let us denote by $\epsilon$ the {\em empty word} (we assume $\epsilon
\not\in \Sigma$). The relation $\Delta$ is a subset of $Q\times
(\Sigma\cup\{\epsilon\}) \times L\times Q$, where $L$ is the {\em
instruction set} $L = \{inc,dec,ifz\}\times\{1,\ldots,k\}$.
A {\em configuration} of $\mathcal A$ is a pair $(q,\bar v)$, where
$q\in Q$ and $\bar v$ is the vector of counter values.
A {\em run} of an incrementing counter automaton is an infinite
sequence of configurations such that, for every pair of
consecutive configurations $(q,\bar v), (q',\bar v')$ an
\emph{incrementing transition} $(q,\bar v) \xrightarrow{l,a}_\dagger
(q',\bar v')$ has been taken.
We say that $(q,\bar v) \xrightarrow{l,a}_\dagger (q',\bar v')$
has been taken if there exist $\bar{v}_\dagger, \bar{v}'_\dagger$
such that $\bar v \leq \bar{v}_\dagger$, $(q,\bar
v_\dagger)\xrightarrow{l,a} (q',\bar v'_\dagger)$, and
$\bar{v}'_\dagger \leq \bar{v}'$, where $(q,\bar v)
\xrightarrow{l,a} (q',\bar v')$
iff $(q, a, l, q') \in \Delta$ and if $l=(inc,i)$
(resp., $(dec,i)$, $(ifz,i)$), then $v'_i = v_i +1$
(resp., $v'_i = v_i - 1$, $v'_i = v_i = 0$) (the ordering
$\bar v \leq \bar v'$ is defined component-wise in
the obvious way).
Notice that once an incrementing transition $(q,\bar v)
\xrightarrow{l,a}_\dagger (q',\bar v')$ has been taken,
counter values may have been increased nondeterministically
before or after the execution of the basic transition
$(q,\bar v) \xrightarrow{l,a} (q',\bar v')$ by an
arbitrary natural number.
We say that an infinite run of $\mathcal A$ over an $\omega$-word
$w\in\Sigma^{\omega}$ is {\em accepting} iff it traverses
a state in $F$ infinitely often.
The nonemptiness problem for increasing counter automata is the
problem of deciding whether there exists at least one $\omega$-word
accepted by it.
In Section 6, we will show that when we restrict our attention
to $\mathbb N$-models, the situation becomes slightly different,
as symmetry does not hold anymore.

\begin{lemma}
There exists a reduction from the nonemptiness problem for
incrementing counter automata over $\omega$-words to the
satisfiability problem for $\mathsf{AE}$ over the class of
strongly discrete linear orders.
\end{lemma}

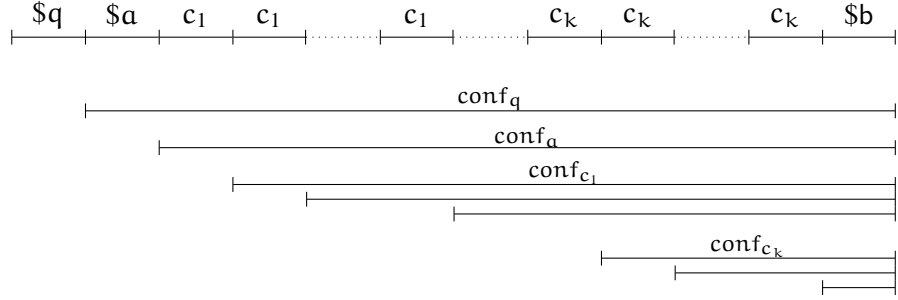
\begin{figure}
\centering
\begin{tikzpicture}[scale=.98]
\draw[|-] (0,0) -- node[above]{$\$q$} ++(1,0)node(end_q){};
\draw[|-] (end_q)++(0,0) -- node[above]{$\$a$} ++(1,0)node(end_a){};
\draw[|-] (end_a)++(0,0) -- node[above]{$c_1$} ++(1,0)node(end_c1_1){};
\draw[|-|] (end_c1_1)++(0,0) -- node[above]{$c_1$} ++(1,0)node(end_c1_2){};
\draw[dotted] (end_c1_2)++(0,0) -- ++(1,0)node(end_c1_3){};
\draw[|-|] (end_c1_3)++(0,0) -- node[above]{$c_1$} ++(1,0)node(end_c1){};
\draw[dotted] (end_c1)++(0,0) -- ++(1,0)node(start_ck){};
\draw[|-] (start_ck)++(0,0) -- node[above]{$c_k$} ++(1,0)node(end_ck_1){};
\draw[|-|] (end_ck_1)++(0,0) -- node[above]{$c_k$} ++(1,0)node(end_ck_2){};
\draw[dotted] (end_ck_2)++(0,0) -- ++(1,0)node(end_ck_3){};
\draw[|-] (end_ck_3)++(0,0) -- node[above]{$c_k$} ++(1,0)node(end_ck){};
\draw[|-|] (end_ck)++(0,0) -- node[above]{$\blank$} ++(1,0)node(end){};
{\footnotesize
\draw[|-|] (end_q)++(0,-1) --node[above=-1mm]{$conf_q$} ++(11,0);
\draw[|-|] (end_a)++(0,-1.5) --node[above=-1mm]{$conf_a$} ++(10,0);
\draw[|-|] (end_c1_1)++(0,-2) --node[above=-1mm]{$conf_{c_1}$} ++(9,0);
\draw[|-|] (end_c1_2)++(0,-2.2) -- ++(8,0);
\draw[|-|] (end_c1)++(0,-2.4) -- ++(6,0);
\draw[|-|] (end_ck_1)++(0,-3) --node[above=-1mm]{$conf_{c_k}$} ++(4,0);
\draw[|-|] (end_ck_2)++(0,-3.2) -- ++(3,0);
\draw[|-|] (end_ck)++(0,-3.4) -- ++(1,0);
}
\end{tikzpicture}
\caption{Encoding of a configuration of an incrementing counter automaton in $\mathsf{AE}$.}
\label{fig:encoding}
\end{figure}

\begin{proof}
Let $\mathcal A = (\Sigma,Q,q_0,C,\Delta,F)$ be an incrementing counter automaton.
We write an $\mathsf{AE}$ formula $\varphi_{\mathcal A}$ which is satisfiable
over the class of strongly discrete linear orders iff there
is at least one $\omega$-word over $\Sigma$ accepted by
$\mathcal A$. Let us assume that $|Q|=\mu+1$, $|\Sigma|=\nu$, $|F|=\eta$,
and $|C|=k$, and there are
\begin{inparaenum}[\it (i)]
\item  $\mu+1$ proposition letters $q_0,q_1,\ldots,q_\mu$, one for
each state in $Q$ ($q_0$ being the initial state);
\item $\nu$ proposition letters $a_1,\ldots,a_\nu$, one for each
symbol in $\Sigma$; and
\item $k$ proposition letters $c_1,\ldots,c_k$, one for each counter
in $C$.
\end{inparaenum}
Moreover, to simplify the formula, we introduce a proposition letter
$\$q$ (resp., $\$a$, $\$c$) which holds at some interval iff at least one
$q_i$ (resp., $a_i$, $c_i$) holds at that interval.
Finally, a proposition letter $conf$ is used to denote a configuration.
Additional auxiliary proposition letters will be introduced later on.

To encode the components of a configuration, we use intervals of the
form $[x,x+1]$ (unit intervals), which are univocally identified
by the $\mathsf{AE}$ formula $\Gfinishes \bot$. A configuration is
modeled by a (non-unit) interval $[x,x+s]$, labeled with $conf$, consisting
of a sequence of unit intervals labeled as follows: $[x,x+1]$ is labeled
with (a proposition letter for) a state in $Q$, $[x+1,x+2]$
by a letter in $\Sigma$, and all the remaining unit intervals, but the
last one (for technical reasons, $[x+s-1,x+s]$ is labeled with a special
proposition letter~\blank), are labeled with counters in $C$.
Figure~\ref{fig:encoding} depicts (part of) the encoding of a configuration.
We constrain any configuration interval $[x,x+s]$ to contain one
unit interval labeled with a state, one labeled with an alphabet letter,
and, for $1 \leq i \leq k$, as many unit intervals labeled with $c_i$
as the value of counter $c_i$ is in that configuration.
Without loss of generality, we can assume all counter values to be initialized
to $0$ ($\bar v=\bar 0$), and thus the initial configuration contains no counter
proposition letters.

Let $[U]\varphi$ be a shorthand for the formula $[U]\varphi = \varphi\wedge
\Gmeets\varphi\wedge\Gmeets\Gmeets\varphi$ ({\em universal} modality).
We first constrain proposition letters that denote states (in $Q$), input
symbols (in $\Sigma$), and counter values to be correctly placed.
\setlength{\partextlength}{5cm}
\begin{align}
&[U](\$q \leftrightarrow \bigvee_{i=0}^{\mu} q_i \wedge
 \$a \leftrightarrow \bigvee_{i=1}^{\nu} a_i \wedge
 \$c \leftrightarrow \bigvee_{i=1}^{k} c_i) \label{form:generic}
 && \partext{placeholders are correctly set}\\
&[U](\Gfinishes \bot \leftrightarrow \$q \vee \$a \vee \$c \vee \blank)\label{form:unit_int}
 && \partext{placeholders are unit intervals} \\
&[U]\bigwedge_{p\in\{q,a,c,b\}}(\$p\rightarrow\neg\bigvee_{p'\in\{q,a,c,b\}, p'\neq p}\$p')
 && \partext{exactly one placeholder per unit interval}\label{form:unique} \\
 &[U](\bigwedge_{i\neq j} (q_i\rightarrow\neg q_j) \wedge
  \bigwedge_{i\neq j} (a_i\rightarrow\neg a_j) \wedge
 \bigwedge_{i\neq j} (c_i\rightarrow\neg c_j)) \label{form:notwo}
 && \partext{exactly one state, letter, counter}
\end{align}
Next, we encode the sequence of configurations as a (unique)
infinite chain that starts at the ending point of the interval
where $\varphi_{\mathcal A}$ is evaluated, and we constrain the
counter values of the initial configuration to be equal to $0$.
To force such a chain to be unique and to prevent configurations
from containing or overlapping other configurations, we introduce
an additional proposition letter $conf'$, which holds over all and only
those intervals which are suffixes of a $conf$-interval.
\setlength{\partextlength}{4.5cm}
\begin{align}
&\meets(conf\wedge \finishes \finishes \top \wedge \Gfinishes\Gfinishes\Gfinishes\bot)\label{form:conf1}
 &&\partext{the initial configuration has two internal points only}\\
&[U](conf\rightarrow\meets conf \wedge \finishes \finishes \top)
 \label{form:conf2}
 &&\partext{a chain of $conf$s; each $conf$ has room for state and letter} \\
&[U]((conf\rightarrow\Gfinishes conf')\wedge(conf'\rightarrow\neg conf))\label{form:conf3}
 &&\partext{$conf$s are ended by $conf'$s which are not $conf$}\\
&[U]\bigl((\meets conf'\rightarrow\neg conf)\wedge
    (conf' \rightarrow \meets conf \wedge \neg \finishes conf)\bigr)
   \label{form:conf4}
 &&\partext{$conf$ neither overlap nor contain other
 $conf$s; $conf'$s end $conf$s}
\end{align}
Now, we force configurations to be properly structured: they must start
with a unit interval labeled with a state (the initial configuration
with $q_0$), followed by a unit interval labeled with an input letter,
possibly followed by a number of unit intervals labeled with counters,
followed by a last unit interval labeled with \blank.
As modalities $\meets$ and $\finishes$ do not allow one,
in general, to refer to the subintervals of a given interval,
a little technical detour is necessary. We introduce the
auxiliary proposition letters $conf_q$, $conf_a$, and $conf_{c_i}$
(one for each type of counter), and we label the suffix of a
configuration interval met by a unit interval labeled with $\$q$
(resp., $\$a$, $c_i$) with $conf_q$ (resp., $conf_a$, $conf_{c_i}$).
In such a way, modality $\finishes$ can be exploited to get an
indirect access to the components of a configuration.
As an example, we use it to force every configuration to
include at most one state and one input letter.
Notice that proposition letter \blank plays an essential role here:
it allows us to associate the last $c_i$ of each configuration
with the corresponding $conf_{c_i}$.
\setlength{\partextlength}{3.75cm}
\begin{align}
& \meets q_0\wedge[U] (\meets conf \leftrightarrow \meets \$q) \label{form:confelem1}
 &&\partext{$conf$ starts with state (the initial one with $q_0$)}\\
& [U]((\$q \rightarrow \meets \$a) \wedge (\$a \vee \$c \rightarrow \meets (\$c \vee \blank))
 \wedge (\blank \rightarrow \meets \$q) )
 &&\partext{$conf$ is properly structured}\\
&[U]((\$q \rightarrow \Gmeets(conf'\rightarrow conf_q))\wedge (\$a
\rightarrow \Gmeets(conf'\rightarrow conf_a)))\label{form:confelem2}
 &&\partext{$\$q$ meets $conf_q$, $\$a$ meets $conf_a$}\\
&[U](\neg (conf_q\wedge\finishes conf_q)\wedge\neg(conf_a\wedge\finishes conf_a))\label{form:confelem3}
 &&\partext{at most one state and one letter per $conf$}\\
&[U](\bigwedge_{i=1}^{k}(c_i\rightarrow\Gmeets(conf'\rightarrow conf_{c_i})))
 &&\partext{$c_i$ meets $conf_{c_i}$}\label{form:confelem4}
\end{align}
To model decrements and increments, auxiliary proposition letters
$c_{dec}, c_{new}, conf_{dec},$ and $conf_{new}$ are introduced.
$c_{dec}$, which labels at most one unit interval $c_i$ of a given
configuration, constrains the value of the $i$-th counter to be
decremented by $1$ by the next transition, provided that $\Delta$
contains such a transition. Similarly, we constrain $c_{new}$ to
label a (unique) unit interval $c_i$ added by the last transition
to represent an increment by $1$ of the value of the $i$-th
counter, provided that $\Delta$ contains such a transition.
\setlength{\partextlength}{3.5cm}
\begin{align}
&[U]\bigl(\bigwedge_{l\in\{new,dec\}}(c_l
\rightarrow(\$c\wedge\Gmeets(conf'\rightarrow conf_l)))\bigr)
 &&\partext{if $c_{l}$, then $conf_{l}$}\\
&[U]\bigl(\bigwedge_{l\in\{new,dec\}}((\Gfinishes \bot
\wedge\meets conf_{l})\rightarrow c_{l})\bigr)
 &&\partext{if $conf_{l}$, then $c_{l}$}\\
&[U](\neg (conf_{dec}\wedge\finishes conf_{dec}) \wedge
\neg(conf_{new}\wedge\finishes conf_{new}))
\label{form:newincdec4}
 && \partext{at most one $conf_{l}$ per $conf$}
\end{align}
To constrain the values that counters may assume in consecutive
configurations, we introduce three auxiliary proposition letters
$corr$, $corr'$, and $corr_{conf}$. To model the faulty behavior of
$\mathcal A$, that can increment, but not decrement, the values
of counters non-deterministically, we allow two $corr$-intervals
to start, but not to end, at the same point.
\setlength{\partextlength}{6.5cm}
\begin{align}
&\Gmeets(\meets c_{new} \rightarrow\neg\finishes corr)\label{form:corr1}
 &&\partext{$new$ counters have not a counterpart in previous $conf$}\\
&[U]((\$q\vee\$a\vee c_{dec})\rightarrow\Gmeets\neg corr)\label{form:corr3}
 &&\partext{$q$s, $a$s, and $dec$ counters have not a counterpart in next $conf$}\\
&[U]((\$c\wedge\neg c_{dec})\rightarrow\meets corr)\label{form:corr2}
 &&\partext{non $dec$ counters have a counterpart in next $conf$} \\
&[U]( (\Gfinishes \bot \wedge \meets corr) \rightarrow \$c )\label{form:corr4}
 &&\partext{$corr$ are met by a counter}\\
\begin{split}&[U]((corr \rightarrow \Gfinishes corr' \wedge \meets \$c)\wedge \\
&\quad \wedge (\meets conf\rightarrow\Gmeets(corr'\rightarrow corr_{conf})))\end{split}\label{form:corr5}
&&\partext{$corr$s are ended by $corr'$s and meet a counter, some $corr'$s are $corr_{conf}$s}\\
\begin{split}&[U](\neg (corr_{conf}\wedge\finishes corr_{conf}) \wedge\\
&\quad\wedge(corr\rightarrow\finishes corr_{conf}))\end{split}\label{form:corr7}
&&\partext{$corr$ connects counters of consecutive $conf$}\\
&[U](\meets corr_{conf} \rightarrow \meets conf)\label{form:corr6}
 &&\partext{$corr_{conf}$ begins $conf$}\\
&[U](\bigwedge_{i=1}^{k}(c_i\rightarrow\Gmeets(corr\rightarrow \meets c_i)))\label{form:corr8}
 &&\partext{each $corr$ corresponds to some counter} \\
&[U] \neg (corr \wedge \finishes corr)\label{form:corr9}
 &&\partext{no $corr$ ends $corr$}
\end{align}
Finally, we constrain consecutive configurations to be related by
some transition $(q,a,l,q')$ in $\Delta$.
%
%
\setlength{\partextlength}{3cm}
\begin{align}
\begin{split}\bigvee_{(q,a,(inc,i),q')\in\Delta}\hspace{-0.25cm}&
  \bigl(\meets(q\wedge\meets a)\wedge\meets(conf\wedge\meets q' \wedge \\[-\baselineskip]
  & \qquad\meets(conf\wedge\finishes (conf_{c_i}\wedge conf_{new})))\bigr)
\end{split}\label{form:deltainc}
&&\partext{instruction $(inc,i)$}
\\
\begin{split}\bigvee_{(q,a,(dec,i),q')\in\Delta}\hspace{-0.25cm}&
 \bigl(\meets(q\wedge\meets a)\wedge\meets(conf\wedge\meets q' \wedge\\[-\baselineskip]
 & \hspace{3cm}\finishes (conf_{c_i}\wedge conf_{dec}))\bigr)
\end{split}\label{form:deltadec}
&&\partext{instruction $(dec,i)$}
\\
\bigvee_{(q,a,(ifz,i),q')\in\Delta}\hspace{-0.25cm}&
 \bigl(\meets(q\wedge\meets a)\wedge\meets(conf\wedge\meets q'\wedge\Gfinishes \neg conf_{c_i})\bigr)
\label{form:deltaifz}
&&\partext{instruction $(ifz,i)$}
\\
[U] \bigl(\meets &conf \rightarrow \bigl((\ref{form:deltainc})\vee(\ref{form:deltadec})\vee(\ref{form:deltaifz})\bigr)\bigr)\label{form:delta3}
&&\partext{an instruction}
\end{align}
We define $\varphi_{\mathcal A}$ as the conjunction of all above formulas
paired with the condition that the infinite computation passes through a final
state infinitely often.
$$
\varphi_{\mathcal A}=\eqref{form:generic}\wedge\ldots\wedge\eqref{form:corr9}\wedge\eqref{form:delta3}\wedge \Gmeets \meets \meets \bigvee_{q_f\in F}q_f
$$
It is straightforward to prove that $\varphi_{\mathcal A}$  is satisfiable iff $\mathcal A$ accepts at least one $\omega$-word.
\end{proof}

\section{NP-Completeness}
\label{sec:NP}

In this section, we prove that NP-completeness of $\mathsf{B
\overline B}$~\cite{Goranko04} can be extended to $\mathsf{B
\overline BL\overline L}$.
Since the satisfiability problem for propositional logic is
NP-complete, every proper fragment of $\mathsf{B\overline BL\overline L}$
including it is at least NP-hard. Unlike the rest of the sections,
the core of this one is a membership proof, namely,
a proof of NP-membership: by a model-theoretic argument,
it shows that satisfiability of a $\mathsf{B\overline BL
\overline L}$-formula $\varphi$ can be reduced to its satisfiability
in a periodic model where the lengths of prefixes and
periods have a bound which is polynomial in $|\varphi|$.

For the sake of simplicity, we consider the case of
$\mathsf{B\overline BL\overline L}$ interpreted over
$\mathbb{N}$. The proof can be generalized to the whole
class of strongly discrete linear orders.
Moreover, it can be shown that satisfiability of a
$\mathsf{B\overline BL\overline L}$-formula $\varphi$ over $\mathbb{N}$
can be reduced to satisfiability of the formula 
$\tau(\varphi)= \later \Tlater \varphi$ over the interval $[0,1]$, that
is, $M,[x,y]\mmodels \varphi$ for some $[x,y]$ if and only $M,[0,1]
\mmodels\tau(\varphi)$.
Thus, we can safely restrict our attention to the problem of
satisfiability over $[0,1]$ ({\em initial} satisfiability).
As a preliminary step, we introduce some useful notation and notions,
including that of periodic model.

\begin{definition}\label{def:periodic-model}
An interval model $M = \langle\mathbb I(\mathbb N), V\rangle$ is
{\em ultimately periodic, with prefix $Pre$ and period $Per$}, if
for every interval $[x,y] \in \mathbb I(\mathbb N)$ and proposition
letter $p \in \mathcal{AP}$,
\begin{inparaenum}[\it (i)]
\item if $x \geq Pre$, then $[x,y] \in V(p)$ iff $[{x+Per}, {y+Per}] \in V(p)$ and
\item if $y \geq {Pre}$, then $[x,y] \in V(p)$ iff $[x,y+Per] \in V(p)$.
\end{inparaenum}
\end{definition}

Let us consider a $\mathsf{B\overline BL\overline L}$-formula $\varphi$.
We define $Cl(\varphi)$ as the set of all subformulas of $\varphi$ and
of their negations.
Let $M$ be a model such that $M,[0,1]\mmodels\varphi$. For every point
$x$ of the model, let $\mathcal R_L(x)$ (resp., $\mathcal R_{\overline
L}(x)$) be the maximal subset of $Cl(\varphi)$ consisting of all and
only those $\later$-formulas (resp., $\Tlater$-formulas) and their
negations that are satisfied over intervals ending (resp., beginning)
at $x$\footnote{It is easy to see that all intervals ending (resp., beginning)
at the same point satisfy the same $\later$-formulas (resp.,
$\Tlater$-formulas).}.
Let $\mathcal R(x)=\mathcal R_L(x)\cup\mathcal R_{\overline L}(x)$.
$\mathcal R(x)$ must be consistent, that is, it cannot contain a
formula and its negation. Let $\mathcal R$ be the subset of
$Cl(\varphi)$ that contains all possible $\later$- and $\Tlater$-formulas.
It is immediate to see that $|\mathcal R| \leq 2|\varphi|$.
In the following, we will also compare intervals with respect to
satisfiability of $\starts$- and $\Tstarts$-formulas.
Given a model $M$, we say that two intervals $[x,y]$ and $[x',y']$
are $B$-equivalent (denoted $[x,y] \equiv_B [x',y']$) if for every
$\starts\psi \in Cl(\varphi)$, $M,[x,y]\mmodels\starts\psi$ iff
$M,[x',y']\mmodels\starts\psi$ and for every $\Tstarts\psi \in
Cl(\varphi)$, $M,[x,y]\mmodels\Tstarts\psi$ iff $M,[x',y']\mmodels
\Tstarts\psi$.
We denote by $m_B$ the number of $\starts$- and $\Tstarts$-formulas
in $Cl(\varphi)$.
To prove that the satisfiability problem for $\mathsf{B\overline
BL\overline L}$ is in NP, we first prove that every satisfiable
formula $\varphi$ has an ultimately periodic model, and then we
show how to possibly contract such a model to obtain a model whose
prefix and period are polynomial in $|\varphi|$.
\begin{lemma}\label{lem:bbll-ultimately-periodic}
Let $\varphi$ be a $\mathsf{B\overline BL\overline L}$-formula
and $M = \langle\mathbb{I}(\mathbb{N}), V\rangle$ be such that
$M, [0,1] \mmodels \varphi$. Then, there exists an
ultimately periodic model $M^{*} = \langle \mathbb{I}(\mathbb{N}),
V^{*}\rangle$ that satisfies $\varphi$.
\end{lemma}
\begin{proof}
Let $M = \langle\mathbb{I}(\mathbb{N}), V\rangle$ be such that
$M, [0,1] \mmodels \varphi$. If $M$ is not ultimately periodic,
we turn it into an ultimately periodic model as follows. First,
by transitivity of $\later$ and $\Tlater$, there must exist
a point $\bar x>1$ such that $\mathcal{R}(y) = \mathcal{R}(\bar x)$
for every $y \geq \bar x$. We take $\bar x$ as the prefix $Pre$.
Then, we take as the period of the model a value $Per > m_B$ that
satisfies the following conditions:
\begin{inparaenum}[\it (i)]
\item for every point $x \leq Pre$ and formula $\later\psi
\in \mathcal{R}(x)$, there exists an interval $[x_\psi,y_\psi]$
such that $M,[x_\psi,y_\psi]\mmodels\psi$ and $x < x_\psi <
y_\psi < Pre + Per$;
\item for every interval $[x,y]$ such that $x < Pre$ and
$y \geq Pre + Per$ and every formula $\Tstarts\psi$ such that
$M,[x,y]\mmodels\Tstarts\psi$, there exists an interval
$[x,y_\psi]$ such that $[x,y]\equiv_B [x,y_\psi]$,
$M,[x,y_\psi]\mmodels\psi$, and $Pre\le y_\psi < Pre+Per$.
\end{inparaenum}
The existence of such a $Per$ is guaranteed by transitivity of
$\starts$ and $\Tstarts$. To force the model to be
periodic, the following additional condition is necessary:
\begin{inparaitem}[\it (iii)]
\item for every interval $[x,y]$ such that $Pre \leq x < Pre+Per$
and $y \geq Pre + 2Per$ and every formula $\Tstarts\psi$ such
that $M,[x,y]\mmodels\Tstarts\psi$, there exists an interval
$[x,y_\psi]$ such that $[x,y]\equiv_B [x,y_\psi]$, $M,[x,y_\psi]
\mmodels\psi$, and $y_\psi < Pre+2Per$.
\end{inparaitem}
If this is not the case, we can change the valuation $V$ to force
condition \emph{(iii)} to be satisfied as follows.
Let $[x,y]$ be an interval that does not satisfy condition \emph{(iii)}.
We choose a finite set of ``witness points'' $\{y_1 < \ldots < y_k\}$
such that (a) for every interval $[x,y']$ and every formula $\starts\psi$,
if $M,[x,y']\mmodels\starts\psi$, then there exists a witness point
$x < y_i < y'$ such that $M,[x,y_i]\mmodels\psi$, and (b) for every
interval $[x,y'']$ and every formula $\Tstarts\theta$, if $M,[x,y'']
\mmodels\Tstarts\theta$, then there exists a witness point $y_j$
such that $M,[x,y_j]\mmodels\psi$ and either $y_j > y''$ or $[x,y_j]
\equiv_B [x,y'']$.
By transitivity of $\starts$ and $\Tstarts$, it follows that
the number of witness points is less than or equal to $m_B$
(the number of $\starts$- and $\Tstarts$-formulas in
$Cl(\varphi)$).

We concentrate our attention on those witness points $\{y_j <
\ldots < y_k\}$ that are greater than $Pre + Per$, and we
turn $V$ into a new valuation $V'$ such that all intervals
starting at $x$ fulfills condition \emph{(iii)} as follows:
\begin{inparaenum}[\it (1)]
\item for every $p \in \mathcal{AP}$ and every $x < y' \leq Pre+Per$,
we put $[x,y'] \in V'(p)$ iff $[x,y']\in V(p)$;
\item for every $p \in \mathcal{AP}$ and every $j \leq i \leq k$,
we put $[x,Pre+Per+i] \in V'(p)$ iff $[x,y_i]\in V(p)$;
\item for every $p \in \mathcal{AP}$ and every $Pre+Per+k < y'
\leq y_k$, we put $[x,y'] \in V'(p)$ iff $[x,y_k]\in V(p)$;
\item the valuation of all other intervals remains unchanged.
\end{inparaenum}
Once such a rewriting has been completed, no other interval $[x,y']$
starting at $x$ can falsify property \emph{(iii)}. By repeating such
a procedure a sufficient number of times, we obtain a model for $\varphi$
that satisfies all the required properties (notice that
properties (1) and (2) are not affected by the rewriting).

The ultimately periodic model $M^{*} = \langle \mathbb{I}(\mathbb{N}),
V^{*}\rangle$ can be built as follows. First, we define the valuation
function $V^*$ for some intervals in the prefix and some intervals in
the first occurrence of the period:
\begin{inparaenum}[\it (1)]
\item for every $p \in \mathcal{AP}$ and every $[x,y]$ such that
$y < Pre+Per$, $[x,y]\in V^*(p)$ iff $[x,y]\in V'(p)$;
\item for every $p \in \mathcal{AP}$ and every $[x,y]$ such that
$Pre \leq x < Pre+Per$ and $y \leq x + Per$, $[x,y]\in V^*(p)$ iff $[x,y]\in V'(p)$.
\end{inparaenum}
Then, we extend $V^*$ to cover the entire model:
\begin{inparaenum}[\it (1)]
\item for every $p \in \mathcal{AP}$ and every $[x,y]$ such
that $x < Pre$ and $y \geq Pre+Per$, $[x,y]\in V^*(p)$ iff $[x,y-Per]\in V^*(p)$;
\item for every $p \in \mathcal{AP}$ and every $[x,y]$
such that $Pre \leq x < Pre+Per$ and $y > x+Per$, $[x,y]\in V^*(p)$
iff $[x,y-Per]\in V^*(p)$;
\item for every $p \in \mathcal{AP}$ and every $[x,y]$ such
that $x \geq Pre+Per$, $[x,y]\in V^*(p)$ iff $[x-Per,y-Per]\in V^*(p)$.
\end{inparaenum}
It is straightforward to prove that $M^*,[0,1] \mmodels \varphi$, and
thus  $M^{*}$ is the ultimately periodic model we were looking for.
\end{proof}
By applying a point-elimination technique similar to the one used
in~\cite{ecai2012} to show NP-membership of $\mathsf{B\overline
BL\overline L}$ over finite linear orders, we can reduce the
length of the prefix and the period of an ultimately periodic
model to a size polynomial in $|\varphi|$, as proved by
the following lemma.
\begin{lemma}\label{lem:bbll-bounded-periodic}
Let $\varphi$ be a $\mathsf{B\overline BL\overline L}$-formula.
Then, $\varphi$ is initially satisfiable over $\mathbb{N}$
iff it is initially satisfiable over an ultimately periodic model $M
= \langle\mathbb{I}(\mathbb{N}),V\rangle$, with prefix $Pre$ and period
$Per$, such that $Pre + Per \le (m_L+2) \cdot m_B + m_L +4$, where
$m_L=2|\mathcal R|$.
\end{lemma}
\begin{proof}
By Lemma~\ref{lem:bbll-ultimately-periodic}, we can assume that
$\varphi$ is initially satisfied over an ultimately periodic model
$M = \langle\mathbb{I}(\mathbb{N}),V\rangle$. If $Pre + Per >
(m_L+2) \cdot m_B + m_L +4$, then we proceed as follows.

\sloppy
Let us consider all points $1<x<Pre+2Per$. For each $\psi\in Cl(\varphi)$
such that $\later\psi\in \mathcal R(x)$ for some $x$ in such a set, we select
$1 < x_{max}^{\psi}\leq Pre+Per$ and $y_{max}^{\psi}<Pre+2Per$ such that
$[x_{max}^{\psi},y_{max}^{\psi}]$ satisfies $\psi$ and for each
$x_{max}^{\psi}<x\leq Pre+Per$ no interval starting at $x$ satisfies
$\psi$. We collect all such points into a set (of $L$-{\em blocked} points)
$Bl_L \subset \{0,\ldots,Pre+2Per\}$.
Then, for each $\psi\in Cl(\varphi)$ such that $\Tlater\psi\in \mathcal R(x)$
for some $1<x<Pre+2Per$, we select an interval $[x_{min}^{\psi},y_{min}^{\psi}]$
that satisfies $\psi$ and such that for each $y<y_{min}^{\psi}$ no interval
ending at $y$ satisfies it.
We collect all points $x_{min}^{\psi}, y_{min}^{\psi}$ into a set
(of $\overline L$-{\em blocked} points) $Bl_{\overline  L} \subset
\{0, \ldots,Pre\}$.
Let $Bl=Bl_L\cup Bl_{\overline  L}\cup\{Pre,Pre+Per\}$. We have that
$|Bl|\le m_L+2$. Now, let us assume $Bl = \{x_1 < x_2 < \ldots < x_n\}$.
For each $0 < i < n$, let $Bl_i=\{x | x_{i} < x < x_{i+1}\}$; similarly,
let $Bl_0=\{x |0 < x < x_1\}$ and $Bl_n=\{x |x_n < x <Pre+2Per\}$.
We prove that if $y,y'\in Bl_i$, for some $i$, then $\mathcal R(y)=
\mathcal R(y')$. The proof is by contradiction. Let us assume $\mathcal R(y)
\neq \mathcal R(y')$. Since $\mathcal R(x)$ is the same for all points $x > Pre$ (it
immediately follows from periodicity), at least one between $y$ and 
$y'$ must belong to the prefix of $M$.
If $\later\psi\in \mathcal R(y)$ and $\later\psi\not\in \mathcal R(y')$, then, by 
definition, $\Glater\neg\psi\in \mathcal R(y')$. This implies that $y<y'$, 
as $\later$ is transitive. It immediately follows that $y < Pre$. 
Let us consider now the above-defined interval $[x_{max}^{\psi},
y_{max}^{\psi}]$. Two cases may arise: either $x_{max}^{\psi} < y$ 
or $x_{max}^{\psi} > y'$. In the former case, since $\later\psi\in \mathcal R(y)$, 
there must exist an interval $[x'',y'']$ satisfying $\psi$ and such 
that $x_{max}^\psi < x'' \leq y'$, thus violating the definition 
of $x_{max}^{\psi}$. 
In the latter case, $\Glater\neg\psi\not\in \mathcal R(y')$, against the 
hypothesis. The case in which $\Tlater\psi\in \mathcal R(y)$ and $\Tlater
\psi\not\in \mathcal R(y')$ can be proved in a similar way. 
Since by assumption $Pre + Per > (m_L+2) \cdot m_B + m_L +4$, by 
a simple combinatorial argument there must exist $x_{i+1} (\leq Pre 
+ Per)$ in $Bl$ such that $|Bl_i|>m_B$. Let $\bar x$ be the smallest
point in $Bl_i$. 
We show that we can build a model $M'=\langle\mathbb{I}(\mathbb{N}
\setminus\{\bar x\}),V'\rangle$, where $\bar x$ has been removed 
and $V'$ is a suitable adaptation of $V$, such that $M',[0,1]
\mmodels\varphi$. 

Let $M''=\langle\mathbb{I}(\mathbb{N}
\setminus\{\bar x\}),V''\rangle$, where $V''$ is the projection 
of $V$ over the intervals that neither start nor end at $\bar x$. 
By definition, replacing $M$ by $M''$ does not affect 
satisfaction of box-formulas (from $Cl(\varphi)$).
The only possible problem is the presence of some diamond-formulas 
which were satisfied in $M$ and are not satisfied anymore in $M''$. 
Let $[x,y]$, with $y<\bar x$, be such that $M,[x,y]\mmodels\later\psi$.
%
%
By definition of $Bl$, there exists an interval $[x_{max}^\psi,y_{max}^\psi]$,
with $x_{max}^\psi,y_{max}^\psi \in Bl$ and $x_{max}^\psi \leq Pre + Per$,
such that $\psi$ holds over $[x_{max}^\psi,y_{max}^\psi]$ and there exists no
interval $[x',y']$, with $x_{max}^\psi < x' \leq Pre + Per$, such that
$\psi$ holds over $[x',y']$. It follows that either $x_{max}^\psi > y$ 
or there exists an interval $[x',y']$ such that $M,[x',y'] \mmodels \psi$ 
and $x' > Pre + Per$. Therefore, $M'',[x,y]\mmodels\later\psi$.
A symmetric argument applies to the case of $\Tlater\psi$. Hence, the removal
of point $\bar x$ does not cause any problem with diamond-formulas of the 
forms $\later \vartheta$ or $\Tlater \vartheta$. 
Assume now that, for some $y<x<\bar x$ (resp., $y<\bar x<x$) and some 
formula $\Tstarts\psi$ (resp., $\starts\psi$) in $Cl(\varphi)$, it is 
the case that $M,[y,x]\mmodels\Tstarts\psi$ (resp., $M,[y,x]\mmodels
\starts\psi$) and that $[y,\bar x]$ was the only interval starting at 
$y$ (in $M$) satisfying $\psi$. Since $\bar x$ is the smallest point in 
$Bl_i$, $M,[y,x_i]\mmodels\Tstarts\psi$ (resp., $M,[y,x_{i+1}]\mmodels
\starts\psi$) by transitivity of $\Tstarts$ (resp., $\starts$). Consider 
now the first $m_B$ successors of $\bar x$: $\bar x+1,\ldots, \bar x+m_B$. 
Since $|Bl_i| > m_B$, we have that all those points belong to $Bl_i$. It 
is possible to prove that there exists a point among them, say, $\bar x+k$, 
that satisfies the following properties:
\begin{inparaenum}[\it (i)]
\item for every $\starts\xi\in Cl(\varphi)$, if $M,[y,\bar x + k + 1]
\mmodels \starts\xi$, then $M,[y,\bar x + k]\mmodels \starts\xi$, and
\item for every $\Tstarts\zeta\in Cl(\varphi)$, if $M,[y,\bar x + k - 1]
\mmodels \Tstarts\zeta$, then $M,[y,\bar x + k]\mmodels \Tstarts\zeta$.
\end{inparaenum}
To prove it, it suffices to observe that, by the transitivity of $\starts$, 
if $M,[y,\bar x + k + 1]\mmodels \starts\xi$ then $M,[y,x']\mmodels \starts\xi$ 
for every $x'\geq\bar x + k +1$. Hence, if $\bar x + k$ does not satisfy 
property \textit{(i)} for $\xi$, all its successors are forced to satisfy 
it for $\xi$. Symmetrically, by the transitivity of $\Tstarts$, if $M,[y,
\bar x + k - 1]\mmodels \Tstarts\zeta$, but $M,[y,\bar x + k]\not\mmodels 
\Tstarts\zeta$, then  $M,[y,x']\not\mmodels \Tstarts\zeta$ for every $x'\geq
\bar x + k$. Hence, all successors of $\bar x + k$ trivially satisfy 
property \textit{(ii)} for $\zeta$. 
Since the number of $\starts$- and $\Tstarts$-formulas is limited by $m_B$, 
a point with the required properties can always be found. We fix the 
defect by defining the labeling $V'$ as follows: we put $[y,\bar x+h]\in V'(p)$ 
if and only if $[y,\bar x+h-1]\in V(p)$, for every proposition letter $p$ 
and $1 \leq t \leq h$. The labeling of the other intervals remain unchanged. 
By definition of $Bl$, it follows that this change in the labeling does not 
introduce any new defect.

By iterating the above-described operation, we obtain an interval model
$\overline{M}=\langle \mathbb{I}(\mathbb{N}), \overline{V}\rangle$, with $Pre + Per \le 
(m_L+2) \cdot m_B + m_L +4$. However, since all changes that we did so 
far are limited to the portion of the model in between $0$ and $Pre+2Per$,
we are not guaranteed that $\overline{M}$ is actually a model for $\varphi$.
To turn it into a model for $\varphi$, we must propagate the changes to 
the rest of the interval model. We proceed as in the proof of Lemma
\ref{lem:bbll-ultimately-periodic}, building an ultimately periodic 
model $M^{*} = \langle \mathbb{I}(\mathbb{N}), V^{*}\rangle$ as follows:
\begin{inparaenum}[\it (i)]
	\item for every $p \in \mathcal{AP}$ and every $[x,y]$ such that $y \leq Pre+Per$, $[x,y]\in V^*(p)$ iff $[x,y]\in \overline{V}(p)$;
	\item for every $p \in \mathcal{AP}$ and every $[x,y]$ such that $Pre < x \leq Pre+Per$ and $y \leq x + Per$, $[x,y]\in V^*(p)$ iff $[x,y]\in \overline{V}(p)$;
	\item for every $p \in \mathcal{AP}$ and every $[x,y]$ such that $x \leq Pre$ and $y > Pre+Per$, $[x,y]\in V^*(p)$ iff $[x,y-Per]\in V^*(p)$;
	\item for every $p \in \mathcal{AP}$ and every $[x,y]$ such that $Pre < x \leq Pre+Per$ and $y > x+Per$, $[x,y]\in V^*(p)$ iff $[x,y-Per]\in V^*(p)$;
	\item for every $p \in \mathcal{AP}$ and every $[x,y]$ such that $x \geq Pre+Per$, $[x,y]\in V^*(p)$ iff $[x-Per,y-Per]\in V^*(p)$.
\end{inparaenum} This concludes the proof.
\end{proof}

%
%
%

\section{NEXPTIME- and EXPSPACE-Completeness}
\label{sec:NEXPEXSP}
The cases of NEXPTIME-complete and EXPSPACE-complete fragments
have been already fully worked out. In the following, we briefly
summarize them.
NEXPTIME-membership of $\mathsf{A\overline A}$ has been proved
in~\cite{Bresolin08b}, while NEXPTIME-hardness of $\mathsf{A}$
over $\mathbb N$ has been shown in~\cite{bresolin07b}.
It is immediate to show that the latter result holds also
for the class of strongly discrete linear orders; moreover,
it can be easily adapted to the case of $\mathsf{\overline A}$,
thus proving NEXPTIME-hardness of any \hs fragment featuring
$\meets$ or $\Tmeets$.
As for EXPSPACE-complete fragments, we know from~\cite{ijfcs2012} that
$\mathsf{AB\overline B\overline L}$ is EXPSPACE-complete.
In~\cite{DBLP:conf/stacs/MontanariSS10}, Montanari et al.\ prove
EXPSPACE-hardness of the fragment $\mathsf{AB}$ over $\mathbb N$
by a reduction from the exponential-corridor tiling problem, which
is known to be EXPSPACE-complete~\cite{convenience_of_tilings}.
The reduction immediately applies to the case of strongly discrete
linear orders. Moreover, it can be easily adapted to the
fragment $\mathsf{A\overline B}$ (a similar adaptation
has been provided for finite linear orders in~\cite{ecai2012}).
Given a tuple $\mathcal T=(T,t_\bot,t_\top,H,$ $V,n)$, where
$T$ is a finite set of tile types, $t_\bot\in T$ is the bottom tile,
$t_\top\in T$ is the top tile, $H$ and $V$ are two binary relations
over $T$, that specify the horizontal and vertical constraints, and
$n \in \mathbb N$, the exponential-corridor tiling problem consists
of deciding whether there exists a tiling function $f$ from a
discrete corridor of height exponential in $n$ to $T$ that associates
the tile $t_\bot$ (resp., $t_\top$) with the bottom (resp., top) row
of the corridor and that satisfies the horizontal and vertical constraints
$H$ and $V$.
The reduction exploits the correspondence between the points inside
the corridor and the intervals of the model.
It makes use of $|T|$ proposition letters to represent the tiling
function $f$; moreover, a binary encoding of each row of the
corridor is provided by means of additional proposition letters; finally,
local constrains on the tiling function $f$ are enforced by using
modalities.

%
%
%

\section{Decidability and Complexity over \texorpdfstring{$\mathbb N$}{N}}

As we already pointed, the asymmetry of $\mathbb N$-models,
which are left-bounded and right-unbounded, is reflected in the
computational behavior of (some of) the fragments of $\mathsf{A\overline AB\overline B}$
and its mirror image $\mathsf{A\overline AE\overline E}$.
More precisely: 
\begin{inparaenum}[\it (i)]
\item $\mathsf{\overline AB}$, but not $\mathsf{AE}$, becomes decidable (non-primitive
recursive)~\cite{DBLP:conf/icalp/MontanariPS10};
\item $\mathsf{\overline A\overline B}$ and $\mathsf{\overline AB\overline B}$, but
not $\mathsf{A\overline E}$ nor $\mathsf{AE\overline E}$, become decidable (this can
be shown by a suitable adaptation of the argument given
in~\cite{DBLP:conf/icalp/MontanariPS10});
\item $\mathsf{\overline ABL}$ and $\mathsf{\overline A\overline BL}$ remain undecidable,
but the proof given in \cite{DBLP:conf/icalp/MontanariPS10} must
be suitably adapted.
\end{inparaenum}

\begin{figure*}[tbp]
   \centering
		\input{nat_dec_fragm}
   \caption{Hasse diagram of all fragments of $\mathsf{A\overline AB\overline B}$ and $\mathsf{A\overline AE\overline E}$ over the natural numbers.}\label{fig:expre_nat}
\end{figure*}

\begin{theorem}
The Hasse diagram in Figure~\ref{fig:expre_nat} correctly shows all the decidable
fragments of \hs\ over $\mathbb N$, their relative expressive power, and the
precise complexity class of their satisfiability problem.
\end{theorem}

The main ingredients of the decidability proof for
$\mathsf{\overline AB\overline B}$ (and thus for
$\mathsf{\overline A\overline B}$ and $\mathsf{\overline AB}$)
can be summarized as follows. Let $\varphi$ be a satisfiable
$\mathsf{\overline AB\overline B}$-formula and let $M =
\langle\mathbb I(\mathbb N), V\rangle$ be a model such that
$M,[x_\varphi,y_\varphi] \mmodels \varphi$ for some interval
$[x_\varphi,y_\varphi]$. It can be easily checked that modalities
$\Tmeets$, $\starts$, and $\Tstarts$ do not allow one to access
any interval $[x,y]$, with $x>x_\varphi$, starting from $[x_\varphi,
y_\varphi]$, and thus valuation over such intervals can be safely
ignored.
%
%
By exploiting such a limitation,
%
%
we can reduce the search for a model of $\varphi$ to a set of
ultimately periodic models only, as it is possible to prove that,
for each satisfiable $\mathsf{\overline  AB\overline B}$-formula,
there exist an ultimately periodic model $M^{*} = \langle
\mathbb{I}(\mathbb{N}), V^{*}\rangle$ and an interval $[x_\varphi,y_\varphi]$
such that $M, [x_\varphi,y_\varphi] \mmodels \varphi$, $y_\varphi < Pre$,
and $Per \leq m_B$, where $m_B$ is the number of $\starts$- and
$\Tstarts$-formulas in $Cl(\varphi)$.
To guess the non-periodic part of the model, the algorithm for satisfiability
checking of $\mathsf{A\overline AB\overline B}$ formulas over finite
linear orders can be used \cite{DBLP:conf/icalp/MontanariPS10}.
Then, the algorithm for satisfiability checking of $\mathsf{AB\overline B}$
formulas over $\mathbb{N}$~\cite{DBLP:conf/stacs/MontanariSS10}
can be applied to check whether the guessed prefix can be extended to
a complete model over $\mathbb{I}(\mathbb{N})$ by guessing the valuation
of intervals $[x,y]$ with $x < Pre$ and $Pre \leq y \leq Pre + Per$.
To prove termination of the algorithm, it suffices to observe that
if the guessed prefix is not \emph{minimal} (in the sense
of~\cite{DBLP:conf/icalp/MontanariPS10}), we can shrink it into
a smaller one that satisfies the minimality condition (see Proposition
2 and Figure 3 in~\cite{DBLP:conf/icalp/MontanariPS10}).
Since the number of minimal prefix models is bounded, and so is the length of
the period, we can conclude that the satisfiability
problem for $\mathsf{\overline AB\overline B}$ over $\mathbb N$ is
decidable. Non-primitive recursiveness has been already shown in~\cite{ecai2012}.

In a very similar way, it is not difficult to adapt the reduction given
in~\cite{DBLP:conf/icalp/MontanariPS10} to prove the undecidability
of $\mathsf{\overline ABL}$ and $\mathsf{\overline A\overline BL}$
over $\mathbb N$. In this case, we reduce the structural termination
problem for lossy counter automata~\cite{lossy} to the satisfiability
problem for $\mathsf{\overline ABL}$ and $\mathsf{\overline A\overline BL}$.
Since the universal modality $[U]$ can be expressed in $\mathsf{\overline ABL}$
and $\mathsf{\overline A\overline BL}$ as
$[U]\varphi=\varphi\wedge\Glater(\TGmeets\varphi\wedge\TGmeets\TGmeets\varphi)$,
one can repeat the entire construction from~\cite{DBLP:conf/icalp/MontanariPS10}
to encode an infinite computation of the lossy counter automata, using $\later$
to impose the required properties on final states.

\smallskip

\noindent{\bf Acknowledgments.} We would like to thank the Spanish
MEC projects TIN2009-14372-C03-01 and RYC-2011-07821 (G. Sciavicco), the Icelandic
Research Fund project Processes and Modal Logics number 100048021 (D.
Della Monica), and the Italian PRIN project Innovative and multi-disciplinary approaches
for constraint and preference reasoning (A. Montanari and D. Della Monica).

\bibliographystyle{eptcs}
\bibliography{gandalf2012}
\end{document}